%% file: main.tex
\documentclass[11pt]{article}
\usepackage{fullpage}
\usepackage{orcidlink}
\usepackage{authblk}

\usepackage{natbib} 
\usepackage{booktabs} 
\usepackage[ruled]{algorithm2e} 

\SetAlFnt{\small}
\SetAlCapFnt{\small}
\SetAlCapNameFnt{\small}
\SetAlCapHSkip{0pt}
\IncMargin{-\parindent}
\usepackage{graphicx} 
\usepackage{amsthm,amsmath,bbold,dsfont}
\usepackage{xcolor}
\usepackage{multirow}

\usepackage{tikz}
\usetikzlibrary{shadings}
\usepackage{hyperref}
\usepackage{cleveref}

\newtheorem{definition}{Definition}
\newtheorem{theorem}{Theorem}
\newtheorem{lemma}{Lemma}
\newtheorem{claim}{Claim}
\newtheorem{informal}{Informal Theorem}

\input{macros}

\setcitestyle{authoryear}

\date{}
\title{Online Combinatorial Allocation with Interdependent Values}

\newcommand*{\email}[1]{\href{mailto:#1}{\nolinkurl{#1}} } 
\author{Michal Feldman}
\affil{
    Tel Aviv University, Israel;
    \email{mfeldman@tauex.ac.il}
    \orcidlink{0000-0002-2915-8405}
}
\author{Simon Mauras}
\affil{
    INRIA, FairPlay joint team, France;
    \email{simon.mauras@inria.fr}
    \orcidlink{0000-0003-4080-3118}
}
\author{Divyarthi Mohan}
\affil{
    Boston University, USA;
    \email{dmohan@bu.edu}
    \orcidlink{0000-0002-8671-5714}
}
\author{Rebecca Reiffenhäuser}
\affil{
    University of Amsterdam, Netherlands;
    \email{r.e.m.reiffenhauser@uva.nl}
    \orcidlink{0000-0002-0959-2589}
}

\begin{document}

\begin{titlepage}

\maketitle

\begin{abstract}
We study online combinatorial allocation problems in the secretary setting, under interdependent values.
In the interdependent model, introduced by Milgrom and Weber (1982), each agent possesses a private signal that captures her information about an item for sale, and the value of every agent depends on the signals held by all agents.
Mauras, Mohan, and Reiffenhäuser (2024) were the first to study interdependent values in online settings, providing constant-approximation guarantees for secretary settings, where agents arrive online along with their signals and values, and the goal is to select the agent with the highest value.

In this work, we extend this framework to {\em combinatorial} secretary problems, where agents have interdependent valuations over {\em bundles} of items, introducing additional challenges due to both combinatorial structure and interdependence. 
We provide  $2e$-competitive algorithms for a broad class of valuation functions, including submodular and XOS functions, matching the approximation guarantees in the single-choice secretary setting. 
Furthermore, our results cover the same range of valuation classes for which constant-factor algorithms exist in classical (non-interdependent) secretary settings, while incurring only an additional factor of $2$ due to interdependence. 
Finally, we extend our study to strategic settings, and provide a $4e$-competitive truthful mechanism for online bipartite matching with interdependent valuations, again meeting  the frontier of what is known, even without interdependence.
\end{abstract}

\bigbreak
\renewcommand{\abstractname}{Acknowledgements}
\begin{abstract}
This work was supported in part by the European Research Council (ERC) under the European Union's Horizon 2020 research and innovation program (grant agreement No. 866132), by an Amazon Research Award, by the Israel Science Foundation Breakthrough Program (grant No. 2600/24), by a grant from the TAU Center for AI and Data Science (TAD), and by the National Science Foundation (NSF) CAREER Award (CCF-2441071).
Simon Mauras benefited from the support of the FMJH Program PGMO (P-2024-0034).
\end{abstract}

\thispagestyle{empty}
\end{titlepage}

\input{sections}


\input{biblio} 

\appendix
\input{appendix}

\end{document}

%% file: macros.tex
\newcommand{\E}{\mathds{E}}
\renewcommand{\Pr}{\mathds{P}}
\newcommand{\ind}{\mathds{1}}
\newcommand{\sigs}{\mathbf s}
\newcommand{\vals}{\mathbf v}
\newcommand{\x}{\mathbf x}
\newcommand{\p}{\mathbf p}
\newcommand{\reals}{\mathds R}
\newcommand{\st}{\text{ s.t. }}
\newcommand{\val}[3]{v_{#1}(#2, #3)}

\newcommand{\opt}[2]{{O_{#1}^{#2}}}

%% file: sections.tex
\section{Introduction}

The secretary problem lies at the heart of optimal stopping theory, serving as a paradigmatic problem in online algorithms \citep{Gardner60,Dynkin63,ferguson1989}.
In this setting, elements arrive sequentially in a random order, with their value revealed upon arrival. Upon observing a value, an online algorithm must make an immediate and irrevocable decision: accept the current value and end the game, or continue to the next element, losing the current one forever. The goal is to compete against the maximum value.
This framework captures a wide range of applications, such as hiring, ad auctions, and online resource allocation, where decisions must be made immediately and irrevocably. A celebrated result in this area establishes the existence of an $e$-competitive algorithm for this problem, and this is tight \cite{Dynkin63}.

While the classic secretary problem is well understood, its basic formulation is too simplistic to capture many real-world complexities. Three major challenges are the following:
\begin{enumerate}
    \item {\bf Combinatorial settings}: Many real-life applications involve structured decision-making beyond single-choice problems, requiring solutions for complex allocation problems.
    \item {\bf Interdependent values}: The standard model assumes that arriving values are independent, whereas in many settings, values are correlated across agents.
    \item {\bf Truthful mechanisms}: When dealing with strategic agents, there is a need for mechanisms that incentivize truthful reporting.
\end{enumerate}

To address the first challenge, prior work has generalized the secretary problem to combinatorial assignment settings \citep{KorulaP09, GoelMehta08,Kleinberg05,KesselheimRTV13}.
These models include weighted bipartite matching and combinatorial allocations, where items are offline, and agents arrive online, revealing their combinatorial valuation function upon arrival. A major breakthrough by \citet{KesselheimRTV13} established an $e$-competitive online algorithm even for combinatorial allocation problems with XOS valuations (a strict superset of submodular valuations). For the special case of (weighted) bipartite matching, \citet{Reiffenhauser19} extended this to a truthful $e$-competitive mechanism, addressing challenges (1) and (3) simultaneously.

The second challenge, that of interdependent values, has only recently been addressed. The interdependent value model was introduced in the seminal work of \citet{MilgromWeber82}, where each agent has a private signal about the item, and their valuation depends on the signals of all agents. Interestingly, \citet{InterdepOptStopping} devised a $2e$-competitive algorithm for the single-choice secretary problem with interdependent values, proving that interdependence incurs at most a factor $2$ loss in the competitive ratio. Furthermore, they established a $4e$-competitive truthful mechanism for this setting, addressing challenges (2) and (3) simultaneously.

\paragraph{Our work: Combining all Three Challenges}
In this work, we provide the first study of the secretary problem that addresses all three challenges simultaneously. We explore combinatorial allocations with interdependent values, both in strategic and non-strategic settings. Notably, even the non-strategic case, that of combinatorial settings with interdependent values, has not been previously studied.

\paragraph{Our Model.}
We consider the problem of allocating a set of $m$ items to a set of $n$ agents. Every agent has a combinatorial valuation function, specifying her value for any bundle of items. 
These valuations are \emph{interdependent} in the following sense:
alongside her valuation function, each agent also has a \emph{signal}, given by a single nonnegative number, which captures her personal information or belief about the quality of the items. Each agent's valuation function depends, besides the allocated bundle of items, also on the signal profile of all agents.

A valuation function thus assigns a real value to every tuple of a bundle of items and a signal profile.
It will be convenient to also think of the valuation function as a set function over items, for a fixed signal profile. Similarly, one can think of the valuation function over signal profiles, for a fixed bundle of items.
For valuations over bundles, we consider XOS valuations (a superset of submodular valuations), and unit-demand valuations (inducing a weighted bipartite matching problem). For valuations over signals, we consider XOS over signals and subadditive over signals.

The agents arrive online, in a random order. Upon the arrival of an agent $i$, her signal and valuation function are revealed, and the online algorithm should make an immediate and irrevocable decision about which items (if any) to assign to her. The goal is to be competitive against the maximum social welfare (i.e., the sum of agent values for their bundles), given the entire information. 

Our two key questions are the following:

\vspace{-0.05in}
 \paragraph{Question 1:} 
What is the competitive ratio that can be obtained for combinatorial secretary problems with interdependent valuations?

\paragraph{Question 2:} 
What is the competitive ratio that can be achieved by truthful mechanisms for combinatorial secretary problems with interdependent valuations?

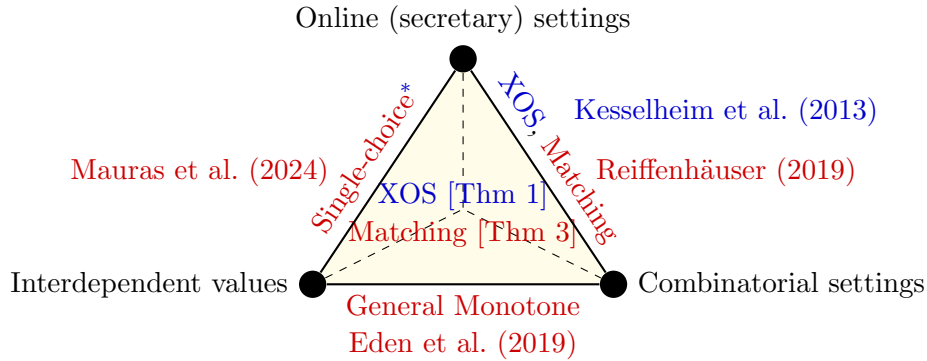
\begin{figure}
    \begin{center}
\begin{tikzpicture}

      \fill[yellow!10] (-2,0) -- (2,0) -- (0,3) -- cycle;

    \node[fill, circle, minimum size = 5pt] (A) at (-2,0) {};
  \node[fill, circle, minimum size = 5pt] (B) at (2,0) {};
  \node[fill, circle, minimum size = 5pt] (C) at (0,3) {};
  \node (O) at (0,1.5) {};

  \draw[dashed] (A) -- (0, 1) -- (B);
  \draw[dashed] (C) -- (0, 1);

    \node[anchor=east] at (A.west) {Interdependent values};
    \node[anchor=west] at (B.east) {Combinatorial settings};
    \node[anchor=south] at (C.north) {Online (secretary) settings};
    \node[anchor=north] at (O) 
    {{\color{blue!80!black}XOS [Thm 1]}};
    \node[anchor=north] at (0,1) 
    {{\color{red!80!black}Matching [Thm 3]}};

    \node at (3.5,2.3) {\textcolor{blue!80!black}{\cite{KesselheimRTV13}}};
    \node at (3.5,1.5) {\textcolor{red!80!black}{\cite{Reiffenhauser19}}};
    \node at (-3.5,1.5) {\textcolor{red!80!black}{\cite{InterdepOptStopping}}};
    \node at (0,-.8) {\textcolor{red!80!black}{\cite{EdenFFGK19}}};
    \draw[thick] (A) -- (B) 
    node[midway, below] {\textcolor{red!80!black}{General Monotone}};
    \draw[thick] (B) -- (C) node[midway, sloped, above] {\textcolor{blue!80!black}{XOS}, \textcolor{red!80!black}{Matching}};
    \draw[thick] (C) -- (A) node[midway, sloped, above] {\textcolor{red!80!black}{Single-choice}{\color{blue!80!black}$^*$}};
\end{tikzpicture}
\end{center}

    \caption{Our work lies at the intersection of three fundamental lines of work: Online (secretary) algorithms, combinatorial auctions, and interdependent values. Works in \textcolor{red!80!black}{red} provide truthful mechanisms, while works in \textcolor{blue!80!black}{blue} provide algorithms for non-strategic settings. \citet{InterdepOptStopping} consider both strategic and non-strategic settings.
    }
    \label{fig:intersection-triangle}
\end{figure}

\subsection{Our Results}

In this paper, we provide answers to both of our questions above, establishing the first constant-competitive algorithms and truthful mechanisms for combinatorial secretary problems with interdependent values. 
Our results are summarized in~\Cref{fig:table-results}, where they are compared to previous work. 
As depicted in Figure~\ref{fig:intersection-triangle}, this work lies at the intersection of three fundamental lines of work: online (secretary) algorithms, combinatorial auctions, and interdependent values.

\begin{table}[h]
    \begin{center} 
    \begin{tabular}{|c|c|c|c|}
        \hline
        Item Valuations & Interdependence & Algorithm & Truthful Mechanism\\ 
        \hline
        \multirow{3}{*}{Single-choice} & None & $e$~\citep{Dynkin63} & $e$~\citep{Dynkin63}\\
        & Subadditive & $2e$~\citep{InterdepOptStopping}  & $4e$~\citep{InterdepOptStopping}  \\  
        & XOS & $4$~\citep{InterdepOptStopping}  & '' \\
         \hline
        \multirow{3}{*}{Bipartite matching}& None & $e$~\citep{KesselheimRTV13} & $e$~\citep{Reiffenhauser19} \\  
         & Subadditive & $2e$~(\Cref{thm:secretary-algo-2e}) & $4e$~(\Cref{thm:secretary-truthful}) \\  
                  & XOS  & $4$~(\Cref{thm:secretary-algo-XOSoS}) & '' \\
        \hline
        \multirow{3}{*}{XOS} & None & $e$~\citep{InterdepOptStopping} & \multirow{3}{*}{\em Unknown} \\  
         & Subadditive & $2e$~(\Cref{thm:secretary-algo-2e}) & \\
         & XOS & $4$~(\Cref{thm:secretary-algo-XOSoS}) & \\
        \hline
    \end{tabular}
\end{center}
    \caption{Competitive ratios of algorithms (non-strategic settings) and truthful mechanisms for combinatorial allocations in the secretary model. Rows represent valuation functions over items, further divided based on valuations over signals.}
    \label{fig:table-results}
\end{table}

For non-strategic settings we show the following:

\begin{informal}
There exists a $2e$-competitive secretary algorithm for combinatorial settings with interdependent valuations that are XOS over items and subadditive over signals.

{Under the stricter condition of XOS over signals (while remaining XOS over items), we improve this to a $4$-competitive algorithm.}
\end{informal}

Notably, we don't impose any additional assumptions on valuations to obtain our results. 
In particular, 
subadditivity over signals is used even for results in the non-combinatorial (single-choice) secretary setting~\citep{InterdepOptStopping}. Furthermore, XOS (over items) is the most general class of valuations for which there are known constant-competitive algorithms in secretary settings, even without interdependence.

For strategic settings we show the following:

\begin{informal}
    There exists a $4e$-competitive truthful mechanism for secretary weighted bipartite matching with interdependent valuations that are separable and subadditive over signals.
\end{informal}

Once again, both assumptions (separability and subadditivity over signals) are made even in (offline) truthful mechanisms for matching with  interdependent valuations~\citep{EdenFFGK19}; that is, no additional assumptions are imposed to obtain our results for the combined setting.

Moreover, in the classic secretary setting, the current frontier for truthful mechanisms in combinatorial problems is edge-weighted bipartite matching, where an $e$-competitive truthful mechanism is known to exist \citep{Reiffenhauser19}. Extending this result to broader combinatorial settings remains a challenging open problem, even without interdependence.

Our results show that secretary assignment problems with interdependent values extend to the same level of generality as their traditional, non-interdependent counterparts, with small loss in the competitive ratio.
This is again demonstrated nicely in Table~\ref{fig:table-results}.
In fact, also for other classes of secretary algorithms, there exist such counterparts that accommodate interdependent values, as outlined in \Cref{sec:framework}.

Furthermore, while the competitive ratios we obtain fall slightly short of those in the non-interdependent setting, they exactly match those established for the single-choice case (with interdependent values) \citep{InterdepOptStopping}. This mirrors the non-interdependent landscape, where competitive ratios for single-choice settings also carry over to more complex combinatorial settings.

Our work is particularly timely, as central offline results under interdependent values are still emerging \citep{EdenFFGK19, EdenGZ22, LuSZ22, EdenFGMM23, EdenFMM24}, and online versions are only beginning to be explored \citep{InterdepOptStopping}. 
For the first time, we demonstrate that constant-competitive online algorithms can be designed in the presence of interdependent values, even in complex combinatorial settings with strategic agents.

\paragraph{A Natural Approach}
Results on secretary algorithms, closely resembling those for the classic secretary problem, all hinge on utilizing independence in the following ways:
First, all values (or combinatorial valuations) are assumed to be fixed in advance, meaning they are not random variables (and thus independent).
Second, which value/valuation function arrives in each time step is drawn uniformly at random from the set of remaining (fixed) ones.
As a result, at any step $t$, the set of arrived values forms a uniformly random subset of size $t$.
This randomness enables the algorithm to treat the initial part of the sequence as a {\em random sample} of the full valuation set, a key property used in existing solutions.

Leveraging this same fact, we observe that one can \emph{separate} the issue of interdependence from the others.
In addition, any pair of the three challenges above (secretary model, combinatorial problem, interdependent values) have been addressed together by previous research. 
One work, most relevant to ours, integrates interdependent values with the secretary problem, but is limited to allocating a single item. 
A second line of research addresses the combinatorial secretary problem, but without interdependence. 
The third approach focuses on designing truthful mechanisms under the secretary model, and is based on analyzing independent probabilities of allocation, and VCG payments, requiring precise knowledge of each item's value --- information that is unavailable at the time of allocation.

A natural and intuitive  approach is therefore to pursue the following strategy. First, we separate the issue of interdependence by using a sufficiently-sized sample. We then apply techniques from all of the above lines of research to obtain constant-competitive algorithms for a large class of combinatorial secretary problems with interdependent values. This approach, as taken in~\Cref{sec:framework}, loses a factor of $4$ compared to secretary algorithms without interdependent values. However, as expected, it is inherently wasteful: 
Separating value interdependence and the online nature of the problem means also separating the information used to handle each one of them, restricting us to using only part of the instance for each.

\paragraph{Integrated Algorithms with Improved Competitive Ratio}
Aiming to avoid wasting a large part of the instance's information, we design algorithms that use both the signal and the valuation function of each agent, thus developing techniques that address all challenges in a unified way. 
Beyond leading to better competitive ratios, this approach also deepens our understanding of the exact restrictions posed by interdependence on the competitive ratio. 
Moreover, it is technically considerably more challenging than the mere combination of existing techniques.

Introducing value interdependence fundamentally changes the independence properties utilized by classic secretary algorithms.
We overcome these challenges by carefully leveraging the aspects of the original independence properties that ``survive" the introduction of interdependence, ultimately achieving similarly concise and intuitive proofs for our results.
These are fueled, in addition to the above, by the exact choices we make in our algorithms and analysis: while many similar versions of our techniques appear initially promising, they differ in how they handle the fundamental tension between utilizing information and maintaining analyzable (in)dependence properties. 
Identifying the ``right" balance is a key aspect of our contribution.

\subsection{Related Work}
\paragraph{Online Resource Allocation.} The work most closely related is that of~\citet{InterdepOptStopping}, who study optimal stopping problems with interdependent values. To the best of our knowledge, other work on online algorithms under the interdependent values model does not yet exist. The secretary problem, in its most fundamental formulation was studied by~\citet{Dynkin63}, who provided the optimal stopping result and $e$-approximation. There is now a long line of work studying a plethora of extensions and generalizations, including e.g. combinatorial allocations~\citep{KesselheimRTV13,HoeferK17,Reiffenhauser19}, or multi-choice secretary under cardinality or feasibility constraints~\citep{Kleinberg05,BabaioffIK07,babaioff2007knapsack}.

Another important paradigm for online resource allocation is the prophet model, where the arrival order is fixed adversarially but the values are drawn independently from known prior distributions. Starting from the seminal works of \citet{KrengelS77, KrengelS78} and~\citet{SamuelCahn84}, who provided optimal prophet inequalities for the single-choice setting, this model has also been well-studied with extensions to combinatorial settings~\citep[e.g.,][]{DuttingFKL20,FeldmanGL15CombAuctions,KleinbergW12}), variations on arrival order \citep[e.g.,][]{Yan11CorGap,PengTang22OrderSelProphet,ArsenisDK21ConstrOrderProphet} and even exploring the combined the power of both models, via the so-called prophet secretary \citep{EsfandiariHLM17ProphetSec}. A key assumption in almost all works is that the values are independent.

\paragraph{Interdependent Values.} The interdependent values model introduced by~\citet{MilgromWeber82}, building on~\citet{wilson1969communications}, paved the road for a long line of work over the past five decades to tackle the standard assumption in mechanism design and resource allocation that agents have independent private values. Given strong impossibility results in this model for obtaining optimal welfare truthfully~\citep{ausubel1999generalized,DM00,JehieM01,JehielMMZ06,ItoP06,CKK15}, recent works in computer science study approximation mechanisms for natural classes of interdependent valuations (e.g., see~\cite{RoughgardenTC16, ChawlaFK14, EdenFFG18, EdenFFGK19, EdenGZ22,gkatzelis2021prior,ChenEW,EdenFTZ21}). 
While much of the focus is on single item auctions,~\citet{EdenFFGK19} provide the first approximation guarantee for combinatorial auctions for any general monotone valuations (over items). Later works consider other combinatorial settings such as public projects~\citep{CohenFMT23}, or fair division~\citep{BirmpasELR23}.
Notably, these works address two of our three challenges, those of (1) combinatorial problems  and (2) interdependent values together in an offline setting.

\section{Model and Preliminaries}

We consider a setting with $n$ agents and $m$  heterogeneous items, where agents have interdependent values over bundles of items. Specifically, every agent $i\in [n]$ holds some private $s_i \in \mathds R_{\ge 0}$, capturing her personal information about the items, and a monotone valuation function
$$v_i:
2^{[m]}\times\mathds R_{\geq 0}^n \rightarrow \mathds R_{\geq 0},$$
where $\val{i}{X}{\sigs}$ is the value of agent $i$ for a bundle of items $X \subseteq [m]$ under a signal profile $\sigs=(s_1,\ldots,s_n)$.
{We also consider strategic settings, where the agents may misreport their signals. As is standard in the interdependent values model we assume that an agent's signal is private while the valuations function is publicly known (upon the arrival of the agent).}
For every subset $A\subseteq [n]$ we write $\sigs_A = (\ind_{1\in A}\cdot s_1, \dots, \ind_{n\in A}\cdot s_n)$, that is, we replace $s_i$ by $0$ if $i\notin A$.

It would sometimes be convenient to consider the valuation function  as a function over signals, for a fixed bundle, or as a function over bundles, for a fixed signal profile.
That is, for any fixed subset of items $X$, $v_i(X,\cdot)$ is a function over signals, and similarly, for any fixed signal profile $\sigs$, $v_i(\cdot,\sigs)$ is a function over bundles.

\subsection{Complement-Free Valuations over Signals}

We first fix a set of items $T \subseteq [m]$ and consider properties of the function $v(\sigs) = \val{}{T}{\sigs}$, which maps every signal profile to a value.

\begin{definition}[Subadditive over signals]
A valuation function $v(\cdot)$ is \emph{subadditive} over signals, if for any signal profile $\sigs$ and any $X\subseteq [n]$ we have
\[
v(\sigs)\le v(\sigs_X) + v(\sigs_{[n]\setminus X}).
\]
\end{definition}

\begin{definition}[XOS over signals]\label{def:xos-sig}
A function $v(\cdot)$ is \emph{XOS} over signals, if for any distribution $\mathcal A \in \Delta(2^{[n]})$ over subset of agents, and any signal profile $\sigs$, we have
\[
\E_{X \sim \mathcal A}[v(\sigs_X)]\ge v(\sigs) \cdot \min_{i\in [n]}\Pr[i\in X].
\]
\end{definition}

\subsection{Complement-Free Valuations over Items}

We now fix a signal profile $\sigs$ and consider properties of the function $v(X) = v(X, \sigs)$, which maps every bundle of items to a value.

\begin{definition}[Unit demand]
A valuation function $v(\cdot)$ is \emph{unit-demand} over items, if there exists a vector of weights $w\in \mathbb R_{\ge 0}^m$ such that for any bundle $X$ we have
\[
v(X) = \max_{j\in X} w_j.
\]
\end{definition}

\begin{definition}[XOS]\label{def:xos}
A valuation function $v(\cdot)$ is \emph{XOS} over items, if for set $S\subseteq [m]$ and any distribution $\mathcal B \in \Delta(2^{S})$ over bundles, we have \[
v(\mathcal B) := \E_{X\sim \mathcal B}[v(X)]\ge v(S) \cdot \min_{j\in S}\Pr_{X\sim \mathcal B}[j\in X].
\]
\end{definition}

\subsection{Approximation Ratio}

We the study online secretary setting, where at each time step $t\in [n]$ agent $i_t$ arrives and we observe their signal $s_{i_t}$. For convenience, we define $A^t := \{i_1, \dots, i_t\}$ the set of agents arrived so far by time step $t$, and we denote $v_i^t(\cdot) := v_i(\cdot,\sigs_{A^t})$ the valuation of agent $i$ using signals that have been observed at time $t$. At the end of step $t$, we irrevocably allocate some bundle $B^t$ to agent $i_t$, which prevents us from computing the optimal allocation.

Given a set of agents $A\subseteq[n]$ and items $J\subseteq[m]$, we denote by $OPT(A,\mathbf w;J)$ the value of optimal allocation of items in $J$ to agents in $A$, considering the weights $\mathbf w := \{w_i: 2^J\mapsto \mathds R_{\geq 0}\}_{i\in A}$. For convenience, we also denote $OPT_i(A,\mathbf w;J)$ and $OPT_{-i}(A,\mathbf w;J)$ to be respective contribution of agent $i$ and of the others to $OPT(A,\mathbf w;J)$.

For simplicity of notations, we may write $OPT(A,\mathbf w)$ or $OPT(A)$, dropping $J$ and (possibly) $\mathbf w$, for special cases where we consider all items $J = [m]$ and specific weights $\mathbf w = \{v_i(\cdot, \sigs_A)\}_{i\in A}$ defined by evaluating bundles using valuations and signals from $A$. Finally, we denote $OPT := OPT([n])$ the value of the overall optimal, which we seek to approximate with our online algorithms.

Given an algorithm, we denote $ALG$ the value of its allocation, which is a random variable which depends on the random arrival of agents and the (possible) internal randomness of the algorithm. We say that the algorithm is an $\alpha$-approximation with $\alpha \geq 1$ if $\E[ALG] \geq OPT/\alpha$.

\section{Secretary Algorithm when Buyers are XOS, and Subadditive over Signals.}
In this section we provide a $2e$-approximation for combinatorial allocations in the non-strategic setting, when valuation functions are XOS over items and subadditive over signals. 

\subsection{Algorithm for Subadditive over Signals}

Our algorithm is a natural extension of the online algorithm by \cite{KesselheimRTV13} for the classic (non-interdependent) setting, where our decisions at each step are made based on valuations that are only partially informed with the signals so far. This is modeled by setting all yet unknown signals to zero when evaluating the valuation functions of buyers, where subadditivity ensures that a randomly-chosen signal subset will suffice to approximate the \emph{real} valuations over all, including future, signals. In particular, the algorithm skips/samples the first $n/e$ agents, and for $t > n/e$,  on arrival of the $t^\text{th}$ agent $i_t$ she is allocated a subset of items $B_t$ that are available from those she would receive in an optimal allocation $OPT(A^t)$ for the currently known instance. Recall that $A^t$ is the set of agents arrived so far by time step $t$. 

A key insight of the analysis resides in decoupling the randomness that affects the information and the availability of the items. More precisely, to bound the contribution of the $t^\text{th}$ agent, we draw the random variables in the following order:
\begin{enumerate}
    \item We first draw (uniformly at random) the set $A^t \subseteq [n]$ containing the first $t$ agents;
    \item Then, we draw (uniformly at random) the agent $i_t \in A^t$ who arrives last within $A^t$;
    \item Finally, we draw (uniformly at random) the arrival order $\sigma_{t-1}$ of $A^{t-1} = A^t\setminus\{i_t\}$.   
\end{enumerate}

\begin{algorithm}
    Sample the first $\lfloor n/e\rfloor$ agents. Let $J^{\lfloor n/e\rfloor+1} = [m]$.
    
    On the arrival of $t^\text{th}$ agent $i_t$, for $t>\lfloor n/e\rfloor$:

\begin{itemize}
    \item Compute an optimal allocation $OPT(A^t)$ of \emph{all} items to the agents $A^t := \{i_1, \dots, i_t\}$ who have arrived, w.r.t. the valuation functions $\{v_i^t(\cdot) = v_i(\cdot,\sigs_{A^t})\}_{i\in A^t}$, denoting $\opt{i}{t}$ the bundle received by agent $i\in A^t$.
    \item From the available items, allocate $B^t = \opt{i_t}{t}\cap J^t$ to agent $i_t$.
    \item Let $J^{t+1} = J^t \setminus B^t$ denote the remaining items.
\end{itemize}
\caption{$2e$-approximation algorithm for subadditive over signals} \label{algo:2e-secretary}
\end{algorithm}

\begin{theorem}\label{thm:secretary-algo-2e}
    \Cref{algo:2e-secretary} is a $2e$-approximation, in expectation, to the offline optimal allocation with interdependent valuations that are XOS over items and subadditive over signals.
\end{theorem}

\begin{proof}
    Let the $k$ denote the size of the sample set. For each time step $t > k$, we compute the contribution $ALG_t$ of the $t^\text{th}$ agent to value obtained by algorithm, by carefully decomposing the randomness of the arrival order as follows:
\begin{align*}
    ALG_t = \E[v_{i_t}(B^t,\sigs_{A^t})]
    = \E_{A^t}\left[\E_{i_t \in A^t}\left[\E_{\sigma_{t-1}}\left[v_{i_t}\left(B^t,\sigs_{A^t}\right)\mid 
    i_t,A^t\right]\mid A^t\right]\right].
\end{align*}

We first invoke the \emph{XOS over items} property to bound the expected value of the allocated bundle $B^t$ in terms of the value for the set $\opt{i_t}{t}$ from the optimal allocation $OPT(A^t)$ that does not take availability into consideration. First, we condition on $A^t = A \subseteq [n]$ and $i_t = i \in A$, and we denote $O = \opt{i}{t}$ to be $i$'s allocation in $OPT(A^t)$.
By~\Cref{lem:xos_items_prob}, we have that
\[    \E_{\sigma_{t-1}}\left[v_{i}\left(O\cap J^t,\sigs_{A}\right)\mid A^{t-1} = A\setminus\{i\}\right] \ge \frac{k}{t-1}\cdot v_{i}\left(O,\sigs_{A}\right).
\]

This yields,
\begin{equation}\label{eq:xos_inner_bound}
    ALG_t \ge \frac{k}{t-1}\cdot \E_{A^t}\left[\E_{i_t \in A^t}\left[v_{i_t}^t\left(\opt{i_t}{t}\right)\mid A^t\right]\right].
\end{equation}

Next we recall that $i_t$ is a uniformly random agent from the set $A^t$. 
We denote $\opt{i}{*}$ the bundle of $i$ in the overall optimal $OPT([n])$. Then, we lower-bound the optimal value of $OPT(A^t)$ with the value that agents from $A_t$ receive in the restriction of $OPT([n])$.
\begin{align}
    ALG_t &\ge \frac{k}{t-1}\cdot \frac{1}{t}\cdot\E_{A^t}\left[\sum_{i \in A^t} v_i^t(\opt{i}{t})\right] \notag \\
     &\ge \frac{k}{t-1}\cdot\frac{1}{t}\cdot \E_{A^t}\left[\sum_{i \in A^t}v_i^t(\opt{i}{*})\right] \label{eq:xos_middle_bound}
\end{align}
Note that for any agent $i$, the probability of $i\in A^t$ is $t/n$. Moreover, given $i \in A^t$, the rest of the $t-1$ agents are random. Thus we can re-write the above inequality as,
\[
ALG_t \ge \frac{k}{t-1}\cdot\frac{1}{t}\cdot\frac{t}{n}\sum_{i=1}^n \E_{A^t}[v_i^t(\opt{i}{*})\mid i\in A^t] = \frac{k}{n}\cdot\frac{1}{t-1}\cdot\sum_{i=1}^n\alpha_t(i),
\]

where for all agents $i$, we define $\alpha_t(i) = \E_{A^t}[v_i(\opt{i}{*},\sigs_{A^t})\mid i\in A^t]$, the expected value of $v_i$ for their optimal bundle on the signal $\sigs_{A\cup\{i\}}$ for a random set $A\subseteq [n] \setminus \{i\}$ of size $t-1$. 

We bound the expected value of the algorithm, by summing over the contribution of all time steps $t> k$, as follows
\begin{equation}\label{eq:ALG_sum_2e}
    \E[ALG] = \sum_{t=k+1}^n ALG_t \ge \frac{k}{n}\cdot\sum_{i=1}^n \sum_{t=k+1}^n\frac{1}{t-1}\cdot \alpha_t(i).
\end{equation}

By subadditivity over signals of the valuation functions, we note that $\alpha_n(i) \le \alpha_t(i) + \alpha_{n-t}(i)$. Moreover, since $v_i$ is monotone (over the signals) we have $\alpha_{t}(i) \ge {\alpha}_{t-1}(i)$ for all $t$.

Using similar arguments as~\cite{InterdepOptStopping} we obtain the following claim for any such $\alpha_t$'s. We include a proof in the Appendix for completeness.
\begin{claim}\label{claim:alpha_sum}
    Given any $\alpha_1 \le \ldots \le \alpha_n $ such that $\alpha_n \le \alpha_t + \alpha_{n-t}$ for all $t\in [n]$, we have
    \[
    \sum_{t=\lceil n/e\rceil}^n \frac{\alpha_t}{t-1} \ge \frac{\alpha_n}{2} + \alpha_n\cdot\Theta(1/n) 
    \]
\end{claim}
With this claim in hand we have, for all agents $i$,
\[
\sum_{t=\lceil n/e\rceil}^n \frac{1}{t-1}\cdot\alpha_t(i) \ge \frac{\alpha_n(i)}{2} + \alpha_n(i)\cdot\Theta(1/n) = \left(\frac 12 + \Theta(1/n)\right) OPT_i,
\]
where $\alpha_n(n) = v_i(\opt{i}{*},\sigs) = OPT_i$ by definition. By plugging this into~\Cref{eq:ALG_sum_2e} for $k = \lfloor n/e\rfloor$ we have,
\begin{align*}
    \E[ALG] &\ge \frac{\lfloor n/e\rfloor}{n}\sum_{i=1}^n\sum_{t=\lceil n/e\rceil}^n \frac{1}{t-1}\cdot\alpha_t(i)\\
&\ge    \frac{\lfloor n/e\rfloor}{n}\sum_{i=1}^n\left(\frac 12 + \Theta(1/n)\right) OPT_i\\
&\ge \left(\frac{1}{2e} + \Theta(1/n)\right)OPT.
\end{align*}
\end{proof}

\begin{lemma}\label{lem:xos_items_prob}
    Let $k$ be the size of the sample set. 
    Fix any subset of agents $A\subseteq[n]$ of size $\ell \ge k$, a bundle of items $B\subseteq [m]$, and an XOS valuation function $\hat v:2^{[m]}\to \reals_{\ge 0}$. We condition on $A^\ell = A$, that is $A$ is the set of agents that arrive by time step $\ell$. Let $J^{\ell+1}$ denote the random subset of remaining items after step $\ell$. The expected value of $\hat v(B\cap J^{\ell+1})$, over the random arrival order of agents in $S$, is at least $\frac{k}{\ell}\cdot \hat v(B)$. That is,
    \[
    \E_{\sigma_{\ell}}[\hat v(B\cap J^{\ell+1}) \mid  A^{\ell} = A] \ge  \frac{k}{\ell}\cdot \hat v(B)
    \]
\end{lemma}

\begin{proof}
    We will show that for any item $j$ the probability that $j\in J^{\ell}$ is $\frac{k}{\ell - 1}$ for all $\ell\ge k$. This immediately proves the lemma since $\hat v$ is an XOS function and by \Cref{def:xos} we have $\E[\hat v(B \cap J^{\ell+1}) \mid  A^\ell = A] \ge \hat v(B) \cdot \min_{j\in B} \Pr[j \in B\cap J^{\ell +1} \mid A^\ell = A ] = \hat v (B)\cdot \frac{k}{\ell}$.\\
    Fix any item $j$. We prove by induction on $\ell \ge k$ that
    \[
    \forall A \st \lvert A \rvert = \ell , \qquad \Pr_{\sigma_\ell}[j \in J^{\ell+1} \mid  A^\ell = A] = \frac{k}{\ell }.
    \]
    The base case $\ell = k$ holds trivially as all the items are still available in the sample phase. Consider a subset of agents $A$ of size $\ell$. The probability the item $j$ is still not allocated after time step $\ell$ is
    \[
    \Pr[j \in J^{\ell+1} \mid  A^{\ell} = A] = \Pr[ j \in J^\ell \text{ and } j \notin \opt{i_\ell}{\ell}\mid  A^\ell = A].
    \]
    Observe that for any available item $j$ to be allocated in step $\ell$, the agent $i^* \in A^\ell$ whose bundle contains $j$ in the allocation $OPT(A^\ell, \vals^\ell)$ has to be the last amongst $A^\ell$ to arrive (i.e., the $\ell^\text{th}$). Since the arrival order is random, the probability of $i^*$ arriving last amongst $A^\ell$ (given any $A^\ell$) is exactly $1/\ell$. Thus, we have
    \begin{align*}
        \Pr\left[j\in J^{\ell+1}\mid A^\ell = A\right] 
        &= \sum_{i\in A\setminus \{i^*\}} \frac{1}{\ell}\cdot\Pr\left[ j \in J^\ell \mid  A^{\ell - 1} = A\setminus \{i\}\right] 
    \end{align*}
By the induction hypothesis it further holds that $\Pr\left[ j \in J^\ell \mid  A^{\ell - 1} = A\setminus \{i\}\right] = k/(\ell -1)$. This simplifies the above equation to,
\[
\Pr\left[j\in J^{\ell+1}\mid A^\ell = A\right]  = \frac{1}{\ell}\cdot\frac{k}{\ell-1}\cdot(\ell-1) = \frac{k}{\ell}
\]
\end{proof}

\subsection{Improved Bounds for XOS over Signals}
 We present a modified online algorithm that obtains a $4$-approximation to the offline optimal when the agents have interdependent valuations that are XOS over signals. In particular, we modify the sample phase to $n/2$ agents instead of $n/e$.

 \begin{algorithm}
    Sample the first $\lfloor n/2\rfloor$ agents. Let $J^{\lfloor n/2\rfloor + 1} = [m]$.
    
    On the arrival of $t^\text{th}$ agent $i_t$, for $t>\lfloor n/2\rfloor$:

\begin{itemize}
    \item Compute an optimal allocation $OPT(A^t)$ of \emph{all} items to the agents $A^t := \{i_1, \dots, i_t\}$ who have arrived, w.r.t. the valuation functions $\{v_i^t(\cdot) = v_i(\cdot,\sigs_{A^t})\}_{i\in A^t}$, denoting $\opt{i}{t}$ the bundle received by agent $i\in A^t$.
    \item From the available items, allocate $B^t = \opt{i_t}{t}\cap J^t$ to agent $i_t$.
    \item Let $J^{t+1} = J^t \setminus B^t$ denote the remaining items.
\end{itemize}
\caption{$4$-approximation algorithm for XOS over signals} \label{algo:4-secretary}
\end{algorithm}

\begin{theorem}\label{thm:secretary-algo-XOSoS}
    \Cref{algo:4-secretary} is an $4$-approximation to offline optimal for online combinatorial allocation in expectation when the interdependent valuations are XOS over items and XOS over signals.
\end{theorem}

\begin{proof}
    Observe that the only difference between~\Cref{algo:4-secretary} and~\Cref{algo:2e-secretary} is the sample size. Therefore we can reuse much of the analysis of~\Cref{algo:2e-secretary}. In particular, we note that the proof of~\Cref{thm:secretary-algo-2e} does not use subadditivity over signals until~\Cref{eq:ALG_sum_2e}. Hence we have,
    \[
    \E[ALG] \ge \frac{k}{n}\cdot\sum_{i=1}^n\sum_{t=k+1}^n \frac{1}{t-1}\cdot\alpha_t(i), \qquad \text{(\Cref{eq:ALG_sum_2e})}
    \]
    where $\alpha_t(i) = \E_{A^t}[v_i(\opt{i}{*},\sigs_{A^t})\mid i\in A^t]$.

    Next, recall that conditioning on  $i\in A^t$, the other $t-1$ agents are a random subset of size $t-1$. Thus, we consider a random set $X$ drawn uniformly from the subsets of size $t$ which include $i$. We use the property that valuation functions are XOS over signals, as defined in \Cref{def:xos-sig}
    \[
    \alpha_t(i) = \E_{A^t}[v_i(\opt{i}{*},\sigs_{A^t})\mid i\in A^t] = \E_X[v_i(\opt{i}{*},\sigs_X)] \ge \min_{j\in[n]}\Pr[j\in X]\cdot v_i(\opt{i}{*},\sigs)
 =  \frac{t-1}{n-1}\cdot OPT_i.
    \]
    Thus by plugging into~\Cref{eq:ALG_sum_2e} for $k = \lfloor n/2\rfloor$ we have,
    \[
    \E[ALG] \ge \frac{k}{n}\sum_{i=1}^n\sum_{t=k+1}^n\frac{1}{t-1}\left(\frac{t-1}{n-1} OPT_i\right) \ge \frac{k(n-k)}{n(n-1)}OPT \geq \frac{1}{4}OPT
    \]
\end{proof}

\subsection{Framework for $4\alpha$-Competitive Algorithms}\label{sec:framework}

We provide a framework to lift up any $\alpha$-competitive algorithm for classic secretary problems (in the absence of interdependence) to an $4\alpha$-competitive algorithm under interdependent valuations that are subadditive over signals. While this takes a black-box approach to re-purpose existing secretary algorithms in the classic setting, it provides a \emph{worse competitive ratio} when applied directly to our secretary problems with interdependence ($4e$ instead of $2e$). The algorithm essentially samples $n/2$ agents to purely use their signals to estimate ``proxy values'' for other agents. Then, it runs any classic $\alpha$-approximation algorithm as a blackbox on a modified instance comprising of the remaining $n/2$ agents, using their non-interdependent proxy values. 

\begin{algorithm}
        \begin{enumerate}
    \item \textbf{First Sample Phase:}
    
    Skip the first $k_1 = n/2$ agents $\hat A = \{i_1,\ldots,i_{k_1}\}$.

For all agents $i$, define a \emph{proxy valuation function} $w_i = v_i(\sigs_{\hat A\cup\{i\}})$.

    \item \textbf{$\alpha$\texttt{-approx-Subroutine} on Modified Instance:}

    As the rest of the $n/2$ agents arrive online, run the $\alpha$-approximation online algorithm for this instance with $n/2$ agents whose values are given by the proxy valuations $\{w_i\}$.
    \end{enumerate}
   \caption{$4\alpha$-approximation for any secretary problem}
\end{algorithm}

Since the modified instance has no interdependence, we get an $\alpha$-approximation to an optimal solution of the modified instance.
By subadditivity over signals, the proxy values provide a good estimate of the true values in expectation. In fact, the expected optimal of the modified instance is a $4$-approximation to the true optimal (w.r.t to the true values of all $n$ agents). This implies that the expected value obtained by the algorithm is a $4\alpha$-approximation to the true optimal.

We note that this framework can be applied for other combinatorial secretary problems, such as multi-choice secretary under downward closed feasibility constraints (cardinality, matroid, knapsack etc.), with competitive algorithms in the classic non-interdependent setting.

We will see an instantiation of this framework formally in the next section, while also tackling the added challenge of strategic agents, to obtain a truthful $4e$-approximation.

\section{Truthful Mechanism for Bipartite Matching}
In this section we consider the case where the agents are strategic and may misreport their \emph{signals} for their own benefit.\footnote{We make the standard assumption in the interdependent values model that the valuation function will be publicly known when the agent arrives.} We provide a truthful mechanism for online weighted bipartite matching that obtains a $4e$-approximation to the offline optimal, when agents have interdependent valuations that are \emph{separable} and \emph{subadditive over signals}.

\subsection{Incentive Compatibility: Preliminaries}\label{sec:IC-prelims}
A mechanism is a tuple of an allocation rule and a payment rule, where the allocation rule determines the bundle that every agent gets, and the payment rule determines how much they pay.
Such a mechanism is said to be {\em truthful} if it incentivizes agents to report their private valuations truthfully.

For each agent $i$ let $x_i(\sigs)$ denote the (possibly random) subset of items that $i$ receives\footnote{{We note that this notation implicitly depends on the arrival order given the online nature of the mechanism. Moreover, the allocation/payment of agent $i$ may depend on the randomness of the instance and any potential randomness in the choices of the mechanism so far.}}. 
Let $p_i(\sigs) \ge 0$ denote the (expected) payment charged (wlog $p_i = 0$ if $x_i = 0$).

\begin{definition}[EPIC]
   A deterministic mechanism $(\x, \p)$ is \emph{ex-post incentive compatible (EPIC)} if truth-telling is a Nash equilibrium, that is, if for every $i\in [n], \sigs, s'_i$ we have
   \[v_i(x_i(\sigs),\sigs) - p_i(\sigs) \ge v_i(x_i(\sigs_{-i},s'_i),\sigs) - p_i(\sigs_{-i},s'_i) .\]
\end{definition}

We say that a randomized mechanism is EPIC (or \emph{universally truthful}) if it is a distribution over deterministic EPIC mechanisms. {We note that this strong notion of incentive compatibility ensures that each agent reports her signal truthfully even when she is aware of all information about others and the random choices of the mechanism.}

To provide truthful mechanism we make an additional assumption of \emph{separability} on the valuation functions. This assumption is also present in work studying truthful approximation mechanisms for combinatorial auctions in the \emph{offline} interdependent values model~\citep{EdenFFGK19}.

\paragraph{Separable valuations.} We say that a valuation function $v_i$ is separable over signals if there exists $f_i: 2^{[m]}\times \mathcal S_i \to \reals_{\ge 0}$ and $g_i: 2^{[m]}\times \mathcal{S}_{-i} \to \reals_{\ge 0}$ such that 
\[
\val{i}{X}{\sigs} = f_i(X,s_i) + g_i(X,\sigs_{-i})
\]
Although universal truthfulness and separability can be defined for general valuation functions over items, we will assume from this point onward that valuations are unit-demand (the optimal allocation will be a bipartite matching). This assumption is necessary, as no online truthful mechanisms are known to exist beyond bipartite matching even without interdependence.

\subsection{Truthful Mechanism}\label{sec:secretary-mechanism}
Our mechanism combines the random sampling generalized VCG mechanism of~\citet{EdenFFGK19} for offline combinatorial interdependent auctions and the truthful mechanism of~\citet{Reiffenhauser19} for online bipartite matching (without interdependence). In particular, we sample the first $n/2$ agents to define proxy valuations $w_i$, which only depend on $i$'s own signal and the signals of these $n/2$ sample agents. By subadditivity over signals the proxy valuations are a good estimate of the real valuations using all signals (including those arriving in the future). We now consider a modified instance of the problem with only the rest of the $n/2$ agents and their proxy valuations. Observe that the modified instance has no interdependence, so we can run the truthful $e$-approximation mechanism of~\cite{Reiffenhauser19}. We then prove that the expected optimal of the modified instance is a $1/4$ approximation to true optimal. Combining these yields a $4e$-approximation.

We start by restating the result of~\cite{Reiffenhauser19} in \Cref{thm:secretary-e-rei19} along with a description of the corresponding allocation algorithm, which we will be using as a sub-routine.

\begin{theorem}[\cite{Reiffenhauser19}]\label{thm:secretary-e-rei19}
    For any instance of classic (non-interdependent) weighted bipartite matching in the secretary setting,  \texttt{Rei19-Allocation-Subroutine} is an $e$-competitive algorithm.
\end{theorem}

\begin{description}
    \item[\texttt{Rei19-allocation-subroutine}~\citep{Reiffenhauser19}:] Given $n$ agents arriving online with valuations $\{w_i\}_{i=1}^n$.
    \item Sample $k = \lfloor n/e\rfloor$ agents. Let $J^{k+1} = [m]$.

    \item For $t > k$, on the arrival of the $t^\text{th}$ agent $i_t$ with valuation function $w_{i^t}$:
   \begin{itemize}
       \item Compute an optimal allocation $OPT(A^t,\mathbf w;J^{t})$ of the \emph{available} items $J^t$ to the agents in $A^t$ w.r.t.\ the valuation functions $\mathbf w =\{w_i\}_{i\in  A^t}$. 
       \item Allocate $B^t = \{j_t\}$ to $i_t$ if they receive  $j^t$ in $OPT( A^t, \mathbf w;J^{t})$, and $B^t = \emptyset$ otherwise.
       \item Let $J^{t+1} = J^t\setminus B^t$
   \end{itemize}
\end{description}

The following is a useful lemma from~\cite{EdenFFGK19} that bounds the expected value of the modified instance on a random subset of agents. We include a proof in the Appendix for completeness.

\begin{lemma}[\cite{EdenFFGK19}]\label{lem:4-approx-random-sampling}
Given a combinatorial allocation problem with $n$ agents who have interdependent valuations $v_i$ that are subadditive-over signals. For a uniformly random subset $\hat A\subseteq [n] $ of size $n/2$ we have,
\[
\E_{\hat A}\left[OPT([n]\setminus \hat A, \mathbf w, [m])\right] \ge \frac{1}{4}\cdot OPT,
\]
where $OPT([n]\setminus \hat A,\mathbf w, [m])$ is the optimal welfare from allocating to only agents in $[n]\setminus \hat A$ w.r.t to proxy valuation functions $w_i(X,s_i) = v_i(X,\sigs_{\hat A})$.
\end{lemma}

\begin{algorithm}
    \begin{enumerate}
    \item \textbf{First Sample Phase:}
    
    Skip the first $k_1 = \lfloor n/2\rfloor$ agents $\hat A = \{i_1,\ldots,i_{k_1}\}$.

For all agents $i$, define a \emph{proxy valuation function} 
$w_i(\cdot) := v_i(\cdot, \sigs_{\hat A\cup\{i\}})$

    \item \textbf{\texttt{Rei19-Allocation-Subroutine}~\citep{Reiffenhauser19} on Modified Instance:}
    
    Sample the next $k_2 = \lfloor n/2e\rfloor$ agents. Let $J^{k_1+k_2+1} = [m]$.

    For $t > k_1 + k_2$, on the arrival of the $t^\text{th}$ agents $i_t$:
   \begin{itemize}
        \item Elicit signal $s_{i_t}$ and use it to compute the proxy valuation function $w_{i^t}$ of agent $i_t$.
       \item Compute an optimal allocation $OPT(\hat A^t,\mathbf w;J^{t})$ of the \emph{available} items $J^t$ to the agents in $\hat A^t := \{i_{k_1+1},\ldots,i_{t}\}$, w.r.t to the proxy valuation functions $\mathbf w =\{w_i\}_{i\in \hat A^t}$. 
       \item Allocate $B^t = \{j_t\}$ to $i_t$ if they receive  $j^t$ in $OPT( A^t, \mathbf w;J^{t})$, and $B^t = \emptyset$ otherwise.
       \item Let $J^{t+1} = J^t\setminus B^t$
   \end{itemize}
   \item \textbf{Payment Rule:}
   
   At the end, for all $t> k_1+k_2$, charge $i_t$ price $p_t(\sigs) =  OPT(\hat A^{t-1}, \mathbf w ;J^t) -  OPT_{-i^t}(\hat A^t,\mathbf w; J^{t}) + g_{i^t}(B^t,\sigs_{[n]-i^t}) - g_{i^t}(B^t, \sigs_{\hat A})$. 
    \end{enumerate}
   \caption{Truthful $4e$-approximation for weighted bipartite matching}
   \label{alg:secretary-mechanism}
\end{algorithm}

We now show that~\Cref{alg:secretary-mechanism} is a truthful mechanism.

\begin{lemma}\label{lem:truthful}
    \Cref{alg:secretary-mechanism} is a universally truthful mechanism. 
\end{lemma}
\begin{proof}
    Fix any arrival order, for simplicity we rename agents so that agent $i$ is the $i^\text{th}$ agent to arrive. Note that all sample agents, i.e, $\hat A = [k_1]$ and $\hat A^{k_1 + k_2} =\{k_1+1 ,\ldots, k_1 + k_2\}$ do not receive any allocation and are charged no payments. More precisely, they have no incentive to misreport as their utility will be $0$ regardless. 
    
    Fix any $i > k_1 + k_2$ with true signal $s_i$, given that all $[n]\setminus i$ agents report $\sigs_{-i}$ truthfully we show that it is utility maximizing for $i$ to report $s_i$ truthfully. 
    For reported signal $s'_i$, denote $w_i'(\cdot) = v_i(\cdot,\sigs_{\hat A},s'_i)$ the corresponding proxy valuation function, and $\mathbf w_{-i} = \{ w_t\}_{t\in \hat A^{i-1}}$ the proxy valuation of agents strictly before $i$ (not part of $\hat A$), which does not depend on agent $i$.
    
    We compute $B^i(s_i')$ the bundle received by $i$ in allocation $OPT(\hat A^i, \mathbf w_{-i}, w_i';J^i)$. Using the separability of the valuation function, the price paid by agent $i$ is then
    \begin{align*}
    \quad p_i(\sigs_{-i},s'_i) &= OPT(\hat A^{i-1}, \mathbf w_{-i};J^i) - OPT_{-i}(\hat A^i,\mathbf w_{-i}, w_i';J^i) &&+ g_i(B^i(s'_i), \sigs_{[n]-i}) - g_i(B^i(s'_i), \sigs_{\hat A})\\
    &&&+ f_i(B^i(s'_i), s_i) - f_i(B^i(s'_i), s_i)\\
    &=OPT(\hat A^{i-1}, \mathbf w_{-i};J^i) - OPT_{-i}(\hat A^i,\mathbf w_{-i}, w_i';J^i) &&+ v_i(B^i(s'_i), \sigs) - v_i(B^i(s'_i), \sigs_{\hat A\cup\{i\}})
    \end{align*}
    With this we are ready to compute $i$'s utility $u_i(\sigs;s'_i)$ for true signals $\sigs$ and signal $s'_i$
    \begin{align*}
       u_i(\sigs;s'_i) &=  v_i(B^i(s'_i),\sigs)  - p_i(\sigs_{-i},s'_i)\\
       &= \underbrace{v_i(B^i(s'_i),\sigs_{\hat A \cup \{i\})}) + OPT_{-i}(\hat A^i,\mathbf w_{-i}, w_i';J^i)}_{\text{Proxy welfare of step } i \text{ for } s'_i} - \underbrace{OPT(\hat A^{i-1}, \mathbf w_{-i};J^i)}_{\text{Does not depend on } i},
    \end{align*}
    where the first set of terms denotes the  proxy welfare obtained by $\hat A^i$ (on true signals) for the allocation $OPT(\hat A^i,\mathbf w_{-i}, w_i';J^i)$. 
    By the definition, this will be at most the proxy welfare for the allocation computed using $i$'s true signal,
    \[
    v_i(B^i(s'_i),\sigs_{\hat A \cup \{i\})}) + OPT_{-i}(\hat A^i,\mathbf w_{-i}, w_i';J^i) \le OPT(\hat A^i, \mathbf w;J^i).
    \]
    Hence we have,
    \[
     u_i(\sigs;s'_i) \le  OPT(\hat A, \mathbf w_{-i} , w_i';J^i) - {OPT(\hat A^{i-1}, \mathbf w_{-i};J^i)}
    \]
   Observe that the right hand side is precisely $i$'s utility $u_i(\sigs;s_i)$ for reporting the true signal $s_i$. Thus proving that the the mechanism is truthful.
\end{proof}

We are now ready to prove the main result of this section, namely that~\Cref{alg:secretary-mechanism} is a $4e$-approximation to the optimal bipartitie matching.

\begin{theorem}\label{thm:secretary-truthful}
    \Cref{alg:secretary-mechanism} is a truthful mechanism that obtains a $4e$-approximation to the optimal offline welfare for online bipartite matching with interdependent valuations that are separable and subadditive over signals.
\end{theorem}

\begin{proof} 
    By \Cref{lem:truthful} we have that our mechanism is truthful for any realization of arrival order. We now show that the expected value of the algorithm is a $4e$-approximation.
    
    Fix any sample $\hat A$ of the first $\lfloor n/2\rfloor$ agents. Conditioned on $\hat A$, we consider the modified instance on $[n] \setminus \hat A$ with valuation functions $w_i$ as defined in \Cref{alg:secretary-mechanism}. We note that this modified instance does not have any interdependence; in particular, the $\lfloor n/2\rfloor$ agents in $[n]\setminus \hat A$ arrive in uniform random order and each agent $i$ arrives with a valuation function $w_i$ that only depends on agent $i$'s signal $s_i$ (along with the signals of the sample $\hat A$ which are not part of the modified instance). This allows us to run the \texttt{Rei19-Allocation-Subroutine}~\citep{Reiffenhauser19} for online bipartite matching in the classic secretary setting (without interdependence). By~\Cref{thm:secretary-e-rei19} this provides an $e$-approximation to the optimal offline allocation---restricted to the modified instance, which implies
    \begin{equation}
        \E\left[\sum_{i\in [n]\setminus \hat A} w_i(B(i))\;\middle|\;\hat A \right] \ge \frac{1}{e}\cdot OPT([n]\setminus\hat A, \mathbf w, [m]) \label{eq:alg-independent-e}
    \end{equation}
    
By subadditivity over the signals, the expected value of $OPT([n]\setminus\hat A, \mathbf w, [m])$ over the randomness of $\hat A$ is at least $OPT/4$, following the result of~\cite{EdenFFGK19}. We cast this in~\Cref{lem:4-approx-random-sampling} and include a proof in the Appendix for completeness.

Putting these together we get the required bound as follows
\begin{align*}
    \E[ALG] = \E\left[\sum_{i=1}^n v_i(B(i),\sigs)\right] & \ge \E_{\hat A}\left[\E\left[\sum_{i=1}^n w_i(B(i)) \;\middle|\;\hat A\right]\right] & \text{(By monotonicity)} \\
    &\ge \E_{\hat A}\left[\frac{1}{e}\cdot OPT([n]\setminus\hat A, \mathbf w, [m])\right] & \text{(By~\Cref{eq:alg-independent-e})} \\
    &\ge \frac{1}{4e}\cdot OPT & \text{(By~\Cref{lem:4-approx-random-sampling})}
\end{align*}
\end{proof}

%% file: biblio.tex

%% file: appendix.tex
\section{Missing Proofs}
\begin{proof}[Proof of \Cref{claim:alpha_sum}] 
We consider three different parts of the sum, (i) from $t_0 = \lceil n/e\rceil$ to $t_1 = \lceil n/(e-1)\rceil$, (ii) from $t_1$ to $t_2 = n-\lceil n/e\rceil + 1$, and (iii) from $t_2$ to $n$. Since $\alpha_t \ge \alpha_{t-1}$ for all $t$, we have
\[
\sum_{t=\lceil n/e\rceil} \frac{\alpha_t}{t-1} \geq 
\alpha_{t_0} \sum_{t=t_0}^{t_1} \frac{1}{t-1} +  
\alpha_{t_1} \sum_{t=t_1+1}^{t_2-1} \frac{1}{t-1} +  
\alpha_{t_2} \sum_{t=t_2}^{n} \frac{1}{t-1}.
\]

We now simplify the sum of $1/(t-1)$ and replace $t_0, t_1, t_2$ with their respective values,
\begin{align*}
\sum_{t=\lceil n/e\rceil} \frac{\alpha_t}{t-1}
&\geq 
\alpha_{t_0} \ln\left(\frac{t_1}{t_0-1}\right) +  
\alpha_{t_1} \ln\left(\frac{t_2-1}{t_1}\right) +  
\alpha_{t_2} \ln\left(\frac{n}{t_2-1}\right)\\
&\geq 
\alpha_{t_0} \ln\left(\frac{e}{e-1}\right) +  
\alpha_{t_1} \ln\left(\frac{(e-1)^2}{e}\right) + 
\alpha_{t_2} \ln\left(\frac{e}{e-1}\right) + \alpha_{t_1}\Theta(1/n)\\
&\geq (\alpha_{t_0} + \alpha_{t_2}) \ln\left(\frac{e}{e-1}\right) + 2\alpha_{t_1} \ln\left(\frac{e-1}{\sqrt{e}}\right) + \alpha_{t_1}\Theta(1/n).
\end{align*}
Finally, observe that $t_2 \ge n - t_0$ and $t_1 \ge n - t_1$. Since $\alpha_n \le \alpha_t + \alpha_{n-t}$ for all $t$ and $\alpha_{t} \ge \alpha_{t'}$ for $t \ge t'$, we bound $\alpha_{t_0} + \alpha_{t_2}$ and $2\alpha_{t_1}$ by $\alpha_n$ to get
\begin{align*}
\sum_{t=\lceil n/e\rceil} \frac{\alpha_t}{t-1}
&\geq \alpha_{n}\ln\left(\frac{e}{e-1}\right) +\alpha_{n}\ln\left(\frac{e-1}{\sqrt{e}}\right) + \alpha_{n}\Theta(1/n)\\
&= \frac{\alpha_n}{2} + \alpha_{n}\Theta(1/n)
\end{align*}
\end{proof}

\begin{proof}[Proof of \Cref{lem:4-approx-random-sampling}]
For any realization of $\hat A$, we define $\hat O_i$ to be $i$'s bundle in the allocation $OPT([n]\setminus \hat A,\mathbf w,[m])$ if $i\in[n]\setminus A$ and $\emptyset$ otherwise.
   \begin{align*}
       \E_{\hat A}\left[OPT([n]\setminus \hat A, \mathbf w,[m])\right] &= \sum_{i=1}^n \E_{\hat A} \left[v_i(\hat O_i, \sigs_{\hat A \cup  \{i\}})\cdot \ind\{i \notin \hat A\}\right] \\
       &\ge \sum_{i=1}^n \E_{\hat A} \left[v_i(\opt{i}{*}, \sigs_{\hat A \cup  \{i\}})\cdot \ind\{i \notin \hat A\}\right] \\
       &= \sum_{i=1}^n \Pr[i\in \hat A] \cdot\E_{\hat A} \left[v_i(\opt{i}{*}, \sigs_{\hat A \cup  \{i\}})\mid{i \notin \hat A}\right]\\
       &\ge \frac{1}{2}\sum_{i=1}^n \frac{1}{2}\cdot v_i(\opt{i}{*}, \sigs), 
   \end{align*}
   where the last inequality follows because $i$ has $1/2$ probability of being chosen in $\hat A$, and by subadditivity over signals we have
   \[
   v_i(\opt{i}{*}, \sigs) \le v_i(\opt{i}{*}, \sigs_{\hat A }) + v_i(\opt{i}{*}, \sigs_{[n] \setminus \hat A}),
   \]
   with probability of $\hat A = A$ and probability of $\hat A = [n]\setminus (A\cup\{i\})$ are equal for any $A\subseteq [n]\setminus\{i\}$ of size $\lfloor n/2 \rfloor$.\footnote{We assume wlog that $n$ is odd by adding a dummy agent with value zero.}
\end{proof}